\documentclass[12pt]{article}
\usepackage{xtab}
\usepackage{float}
\usepackage{framed}
\usepackage{multicol}
\usepackage{graphicx}
\usepackage[labelfont=bf]{caption}
\usepackage{subcaption}
\usepackage{epsfig}
\usepackage{psfrag}
\usepackage[all]{xy}
\usepackage{pdfpages}
\usepackage[blocks]{authblk}

\usepackage{fullpage}
\usepackage{sectsty}
\partfont{\centering}
\usepackage{lmodern}
\usepackage{setspace}
\usepackage[bottom]{footmisc}
\usepackage{longtable}
\usepackage{lscape}
\usepackage{natbib}
\usepackage{amsmath}
\usepackage{amsfonts}
\usepackage{amsthm}
\usepackage{thmtools}
\usepackage{amssymb}
\usepackage{bbm}
\usepackage{array}
\usepackage[mathscr]{euscript}
\usepackage{units}
\usepackage{relsize}
\usepackage{mathtools}
\usepackage{enumitem}
\usepackage[labelsep=period]{caption}

\theoremstyle{plain}
\newtheorem{thm}{Theorem}
\newtheorem{lemma}[thm]{Lemma}
\newtheorem{prop}[thm]{Proposition}
\newtheorem{assumption}[thm]{Assumption}

\makeatletter
\renewenvironment{proof}[1][\proofname] {\par\pushQED{\qed}\normalfont\topsep6\p@\@plus6\p@\relax\trivlist\item[\hskip\labelsep\bfseries#1\@addpunct{:}]\ignorespaces}{\popQED\endtrivlist\@endpefalse}
\makeatother

\usepackage{color}
\definecolor{direct}{HTML}{FF0000}
\definecolor{indirect}{HTML}{FF9999}
\definecolor{dirin}{HTML}{990000}
\definecolor{control}{HTML}{999999}
\definecolor{directE}{HTML}{ED1C24}
\definecolor{indirectE}{HTML}{F69679}
\definecolor{controlE}{HTML}{999999}

\newcommand{\M}{\mathbf{M}}

\renewcommand{\qed}{\space $\square$}

\makeatletter
\newcommand{\oset}[3][0ex]{
  \mathrel{\mathop{#3}\limits^{
    \vbox to#1{\kern-2\ex@
    \hbox{$\scriptstyle#2$}\vss}}}}
\makeatother

\usepackage{mleftright}
\mleftright

\makeatletter
\newcommand{\listintertext}{\@ifstar\listintertext@\listintertext@@}
\newcommand{\listintertext@}[1]{
  \hspace*{-\@totalleftmargin}#1}
\newcommand{\listintertext@@}[1]{
  \hspace{-\leftmargin}#1}
\makeatother

\allowdisplaybreaks[1]

\makeatletter
\renewcommand*{\ext@figure}{lot}
\let\c@figure\c@table
\let\ftype@figure\ftype@table

\makeatother

\usepackage[colorlinks=true, allcolors=black, linktoc=all, hypertexnames=false]{hyperref}

\usepackage[flushleft]{threeparttable}
\usepackage{thmtools}
\usepackage{nameref}
\usepackage{cleveref}
\makeatletter
\newcommand{\setword}[2]{
  \phantomsection
  #1\def\@currentlabel{\unexpanded{#1}}\label{#2}
}
\makeatother

\usepackage{times}

\title{Listwise Deletion in High Dimensions}
\author{J. Sophia Wang and P. M. Aronow\footnote{J. Sophia Wang is Graduate Student, Department of Political Science, Yale University. P. M. Aronow is Associate Professor, Departments of Political Science, Biostatistics, and Statistics and Data Science, Yale University. Contact: peter.aronow@yale.edu.}}
\date{\today}

\begin{document}

\maketitle

\vspace{-0.2in}

\begin{abstract}

We consider the properties of listwise deletion when both $n$ and the number of variables grow large. We show that when (i) all data has some idiosyncratic missingness and (ii) the number of variables grows superlogarithmically in $n$, then, for large $n$, listwise deletion will drop all rows with probability 1. Using two canonical datasets from the study of comparative politics and international relations, we provide numerical illustration that these problems may emerge in real world settings. These results suggest, in practice, using listwise deletion may mean using few of the variables available to the researcher.

\end{abstract}

\section{Introduction}

Listwise deletion is a commonly used approach for addressing missing data that entails excluding any observations that have missing data for any variable used in an analysis. It constitutes the default behavior for standard data analyses in popular software packages: for example, rows with any missing data are by default omitted by the $lm$ function in R \citep{r}, the $regress$ command in Stata \citep{regress}, and the $glmnet$ function in the R package of the same name \citep{glmnet}. 

However, scholars have increasingly recognized that listwise deletion may not be a generally appropriate research method to handle missing data. While a common critique focuses on the plausibility of the ``missing completely at random" assumption (\citeauthor{schafer1997}, \citeyear{schafer1997}, p.~23; \citeauthor{Allison2001}, \citeyear{Allison2001}, p.~6-7, \citeauthor{cameron2005microeconometrics}, \citeyear{cameron2005microeconometrics}, p.~928, \citeauthor{little2019statistical}, \citeyear{little2019statistical}, p.~15), issues about efficiency in estimation have also been raised  (\citeauthor{berk1983handbook}, \citeyear{berk1983handbook}, p.~540,  \citeauthor{schafer1997}, \citeyear{schafer1997}, p.~38, \citeauthor{Allison2001}, \citeyear{Allison2001}, p.~6).
Namely, since listwise deletion discards data, the resulting estimators can be inefficient relative to approaches that use more of the data (e.g., imputation methods). 

These issues have been raised to an audience of political scientists \citep{king2001analyzing,lall2016multiple,honaker2010missing}, but the manner in which listwise deletion can hinder the researcher has been underappreciated.  Namely, if the researcher seeks to use many variables with missingness, it may be impossible altogether to draw any statistical conclusion whatsoever. Accordingly, the use of listwise deletion may imply a severe restriction on variables used in an analysis. 

The primary purpose of this note is to make this argument rigorous by considering the properties of listwise deletion when {\it both} the number of variables $k$ and the number of units $n$ are large. We show that when (i) all variables have some idiosyncratic missingness and (ii) the number of variables grows with $n$ at any superlogarithmic rate, listwise deletion will yield no usable data asymptotically with probability 1. In Supporting Information A, we report numerical illustrations to shed light on finite-$n$ properties under our assumptions. 

We then demonstrate real-world implications by considering two real-world datasets: the Quality of Government dataset \citep{qog2021} and the State Failures dataset \citep{king2007dataverse,kingzeng}. We first report on the empirical patterns of missingness in these datasets. We then conduct a simulation study by randomly subsampling from the variables in these datasets. We show that, even when a qualitatively small number of variables have been chosen from these datasets, very little if any of the data may remain after listwise deletion. Taken together, we conclude that listwise deletion is simply not viable in many data-analytic settings. 


\section{Theory}

We consider a fairly general setting designed to accommodate probabilistic missingness in data. Our results will apply to any estimator, algorithm, or procedure (including, e.g., variable selection or regression) on datasets in this setting, so long as this researcher's chosen procedure depends on listwise deletion.

Before we proceed, we introduce some notation. Let $n$ be the number of observations. Let $k$ be the number of variables (columns) in the dataset. Let $M_{ij}$ be a random indicator variable for whether or not the $j$th variable in the $i$th row is missing. For notational convenience, we will let $\M_{ij}$ represent the random vector collecting the missingness indicators up to variable $j$, $(M_{i1},M_{i2},\ldots,M_{ij})$. 

We will invoke three key assumptions for our results. These assumptions are fully general with respect to standard assumptions about missingness --- i.e., they are compatible with both missing-at-random and missing-not-at-random data generating processes. The first of these assumptions is mutual independence of missingness across rows. This would be violated when, e.g., observations are clustered, as would often be the case for longitudinal data. 

\begin{assumption} \label{independence}
All rows of the data $\left( (M_{11},...,M_{1k} ), ..., (M_{n1},...,M_{nk} )  \right) $ are mutually independent.
\end{assumption}

We also assume that there is some idiosyncratic missingness in each variable. Namely, we will assume that there is a factorization of the data such that all conditional probabilities that an observation is missing are bounded away from zero, given whenever prior variables are observed. Substantively, this assumption rules out the case where, e.g,, one variable is always observed whenever another variable is observed. 

\begin{assumption}\label{qbound}
There exists a $q_* \in [0,1)$ such that for all $i$,
\begin{itemize}
\item $\Pr( M_{i1} = 0) \leq q_*$
\item 
$\Pr( M_{ij} = 0 | \M_{i(j-1)} = \mathbf{0} ) \le q_*$, for all $j \in \{2,\ldots,k\}$ such that $\Pr(\M_{i(j-1)} = \mathbf{0}) > 0$ 
\end{itemize}
\end{assumption}

Note that this is an assumption about the presence of missingness and not about how such missingness relates to outcomes. Thus Assumption \ref{qbound} is compatible with both missing-at-random and missing-not-at-random data generating processes.

With some additional notation, Assumption \ref{qbound} can be weakened to only require the existence of an index ordering over $j$ such that the condition holds. We also note that Assumpton \ref{qbound}  may be unrealistic in settings where missingness is always identical across some groups of variables (e.g., because two or more variables come from a common data source). In such settings, our results can be generalized to the setting where $k$ refers to the number of groups of variables, rather than the number of variables themselves. We formalize this extension in Supporting Information B.

\subsection{Results}

We begin by application of elementary probability theory to yield the following finite-$n$ result. In words, this result establishes an exact lower bound on the probability that all rows of the data will suffer from some missingness. In such instances, listwise deletion would yield no usable data.


\begin{lemma} Under Assumptions \ref{independence} and \ref{qbound}, the probability that listwise deletion removes all rows is $p_{all} \geq (1-q_*^k)^n$. 
\end{lemma}

This result will be helpful in proving our main result shortly in Proposition 5. We can now consider the asymptotic properties of listwise deletion, letting both $k$ and $n$ tend to infinity. To do so, we embed the above problem into a sequence. We let $k_n = f(n)$, where $f$ has range over the natural numbers, and allow $M_{ij,n}$ and therefore $q_{ij,n}$ to vary at each $n$. To ease exposition, we omit notational dependence on $n$.

We have our third and final assumption: $k$ grows superlogarithmically in $n$. This is the primary point of divergence from standard (low-dimensional) theoretical treatments, under which $k$ is assumed to be fixed regardless of $n$. We emphasize here that Assumption 4 -- like other assumptions about asymptotics -- needs not be thought of as a literal growth process that a researcher might follow, but rather an approximation of probabilistic behavior when $n$ (and here, also $k$) are large.\footnote{\citet[p. 255]{lehmann2006elements} provides a good discussion in the context of a sample of size $n$ from a population of size $N$: ``... we must go back to the purpose of embedding a given situation in a fictitious sequence: to obtain a simple and accurate approximation. The embedding sequence is thus an artifice and has only this purpose which is concerned with a particular pair of values $N$ and $n$ and which need not correspond to what we would do in practice as these values change." Our discussion is analogous, with our ``pair of values" being $k$ and $n$. Thanks to Fredrik S\"avje for suggesting the reference.}

Superlogarithmic rates can be extremely slow, and include any polynomial rate of growth (e.g., $n^c$ for any $c>0$). Thus our results can speak to cases where $n$ is large, $k$ is large, but $k \ll n$. To see how slow these rates can be, our results would include the rate $\lfloor{\log(n)^{1.1}}\rfloor$, which would permit use of 2 variables with 10 observations, 8 variables with a thousand observations, and 17 variables with a million observations. This assumption encompasses rates that are slow enough that that they would not normally preclude good asymptotic behavior for most standard estimators. For example, see the minimal assumptions invoked by \citet{lai1979strong} for convergence of the least squares estimator.

\begin{assumption}\label{superlog}

The number of covariates grows superlogarithmically in $n$, so that $\lim_{n\rightarrow\infty} \frac{f(n)}{log(n)} = \infty$. 

\end{assumption}

Assumption \ref{superlog} can be equivalently written in asymptotic shorthand notation as  $k = \omega(\log n)$. The following proposition demonstrates that when Assumptions \ref{independence}, \ref{qbound} and \ref{superlog} hold,  then the probability of listwise deletion yielding no usable data tends to $1$ as $n\rightarrow\infty$.

\begin{prop}\label{prop}  Under Assumptions \ref{independence}, \ref{qbound} and \ref{superlog}, $\lim_{n \to \infty} p_{all} = 1$.
\end{prop}

Thus we have shown that even modest rates of growth in the number of covariates can render any resulting statistical inference asymptotically impossible with listwise deletion. Our results however, critically depend on the assumption that the number of covariates exhibits such growth in $n$, otherwise it is possible that $\lim_{n \to \infty} p_{all} = 0$.

Our results are supported by numerical illustrations in Supporting Information A, which also demonstrate finite-$n$ implications.\footnote{Data and code to replicate all simulations and numerical illustrations are available at \citet{dataverse}.} These results demonstrate that our theoretical results are most relevant in finite-$n$ settings when rates of idiosyncratic missingness are high. When there are low rates of idiosyncratic missingness (e.g., 1\%), the probability that all rows will be removed by listwise deletion can remain extremely low even when $k$ is qualitatively large (e.g., $k=150$) and $n$ is qualitatively small (e.g., $n=100$). However, our results are striking once idiosyncratic missingness rates approach $10\%$ or $25\%$, with striking consequences to the amount of data remaining following listwise deletion.

\section{Application}

In order to understand the real-world operating characteristics of listwise deletion, we considered two prominent datasets in use in the fields of comparative politics and international relations: the January 2021 Quality of Government (QoG) \citep{qog2021} standard cross-sectional dataset, and the State Failures dataset covering country-years from 1955-1998 \citep{king2007dataverse} reported by \citet{sfphase1,sfphase2} and considered by \citet{kingzeng}. Table 1 provides summary statistics on these datasets. We applied mild preprocessing to these datasets: we removed country code variables, and in the case of the State Failure data, to apply the principle of charity, we removed $19$ variables that exhibited 100\% missingness.

\begin{table}[h]
\centering
\caption{Summary Statistics for QoG and State Failures Datasets}
\begin{tabular}{l*{3}{c}}
\hline\hline
& Quality of Government & State Failures \\
\hline
Units of Observation & Countries & Country-Years \\
Number of Observations  & $194$ & $8580$ \\
Number of Variables & $351$ & $1205$ \\
Proportion of Missing Values (Avg.) & $35.9\%$  & $66.8\%$ \\
Proportion of Missing Values (Max) &  $90.7\%$ & $99.99\%$ \\
Number of Variables Fully Observed & $6$ & $79$\\
\hline
\end{tabular}
\end{table}

\subsection{Methodology}

We conducted simulations that ask: how much data is lost by listwise deletion if we randomly subsample $k$ of the variables included in each of these datasets? Our simulation thus attempts to understand statistical behavior when using variables typical (or at least representative) of the major datasets in use in comparative politics and international relations. To do so, we conducted $25,000$ simulations in each of which we drew $k$ of the variables from each dataset (without replacement). We then recorded the number of rows of the data that survive listwise deletion. 

We then report the expected proportion of remaining data observed after listwise deletion, as well as the probability that all rows of the data exhibit some missingness. Note that, here, expected values and probabilities refer to the randomness induced by our random subsampling procedure, not any fundamental stochasticity giving rise to the underlying data. Insofar as random sampling of variables from these datasets codifies a notion of representativeness, interpretation naturally follows from that notion.


We briefly discuss how out theoretical assumptions align with this setting. Assumption 1 asserts i.i.d. missingness across rows, which is unlikely to be met in this setting. In the crossectional QoG case, some countries (e.g., members of the EU) have highly correlated missingness patterns, in part because of shared data availability. In the State Failures dataset, this is more dramatic: since observations are at the country-year, missingness is typically positively correlated within countries. The consequence -- as in other clustering problems -- is that the effective $n$ may be smaller than the nominal $n$ in practice. Our theoretical assumption of i.i.d. missingness across rows can therefore be seen as optimistic when faced with real-world data and, all else equal, the probability that all data will be lost may be higher than theory might dictate.  

Assumption 2 asserts idiosyncratic missingness; i.e., that there exists a factorization of the data such that all observations have some probability of missingness. Assumption 2 holds in our setting, because each observation in each dataset has missingness on at least one variable. Since our subsamples are formed via random sampling, if the first $k-1$ randomly selected variables are not missing, it follows that the $k$th variable has a nonzero probability of being missing. Thus our theoretical assumption of idiosyncratic missingness is compatible with the data, since it is met under a model of random sampling of variables.

\subsection{Results}

\setcounter{figure}{0}    

The left panel in Figure 1 plots the expected proportion of observed data against the number of randomly selected variables in the simulation for the QoG dataset. This proportion monotonically decreases to $0$ as the number of randomly selected variables increases. With only $2$ randomly-selected variables used, the researcher can expect to lose more than $50\%$ of the data, which becomes more than $99\%$ when $17$ randomly-selected variables are used. 

The right panel in Figure 1 plots the probability of all rows will be unusable following listwise deletion against the number of randomly selected variables on the $x$-axis. Our theory would predict that this probability should converge to $1$ as we increase the number of variables selected, and this is borne out in our simulation. The researcher can expect to lose all of the data with probability greater than $0.5$ with 14 variables selected. With $52$ variables used, this probability becomes more than $0.99$. 

\begin{figure}[h]
\centering
\includegraphics[width=6in]{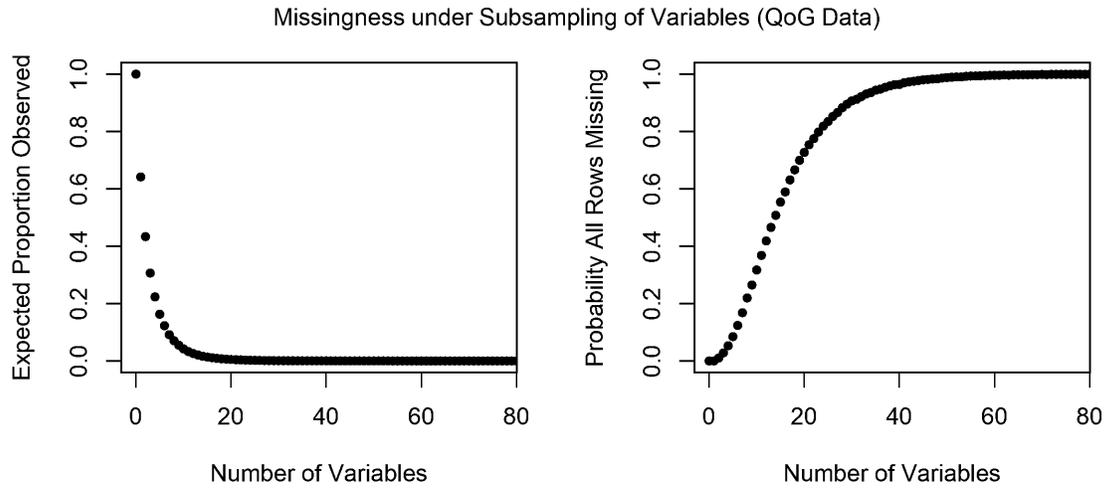}
\caption{Properties of listwise deletion on random subsamples of variables from the QoG dataset.}
\end{figure}

We found similar, but more dramatic, trends with the State Failures dataset. The left panel in Figure 2 plots the expected proportion of observed data against the number of randomly selected variables for the State Failures dataset. The proportion decreases at a faster rate compared to the QOG dataset. The researcher can expect to lose more than $50\%$ of the data if the only one variables included are randomly selected. With more than three variables, the loss is more than $99\%$. 

The right panel in Figure 2 plots the probability of all rows missing under listwise deletion for every number of  randomly selected variables on the $x$-axis. This probability converges to $1$ in a faster rate, as the researcher can expect to lose all of the data with probability greater than $0.5$ with $6$ variables included and the probability rises to be greater than $0.99$ when including more than $17$ variables. Taken together with our results from the QoG data, these results demonstrate that the moral of our theoretical results can be seen in real-world settings.

\begin{figure}[h]
\centering
\includegraphics[width=6in]{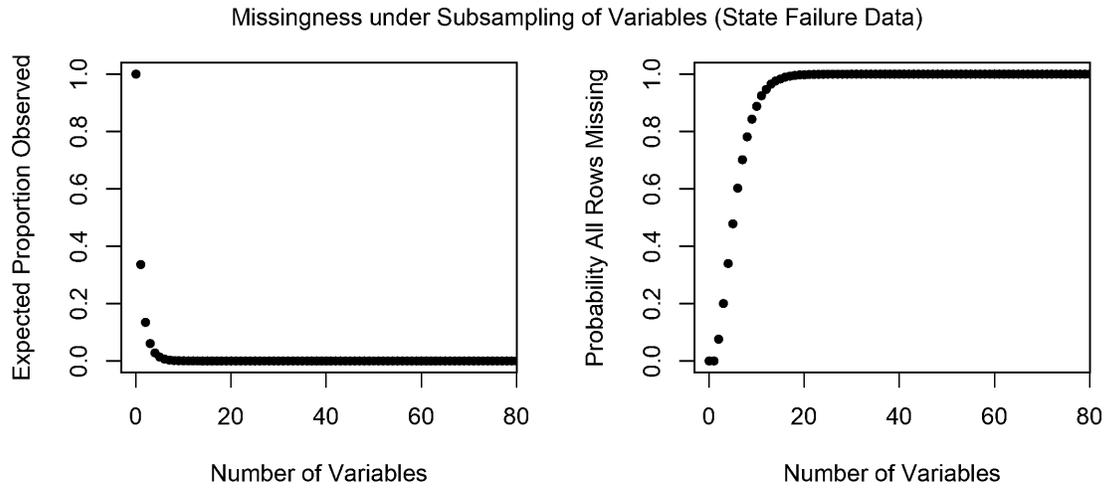}
\caption{Properties of listwise deletion on random subsamples of variables from State Failures dataset.}
\end{figure}

\section{Discussion}

Our results demonstrate that listwise deletion cannot generally accommodate many variables, and that this problem is not resolved asymptotically. Application of high dimensional asymptotics reveals that listwise deletion is even more fragile than was previously understood. Examining real world data used in the fields of comparative politics and international relations highlights the seriousness of these issues for the types of data that political scientists use.

Our results imply that scholars who are committed to listwise deletion may be unable to use all of the variables that are necessary for an otherwise valid data analysis even when $n$ is large. For example, in order to achieve valid inferences in an observational study, a scholar may identify a large number of variables necessary to be conditioned on. However, if these variables exhibit idiosyncratic missingness, then the use of listwise deletion would require the scholar to exclude variables that would be necessary to attain an unbiased estimate. Neither dropping necessary variables nor dropping many observations is desirable. Approaches that avoid listwise deletion exist, including in the high-dimensional setting \citep[e.g.,][]{liu}, and the researcher should consider these alternatives.

We conclude by emphasizing that this note should not be read as advocacy for the generic use of any particular method for addressing missing data. As \citet{arel2018can} and  \citet{pepinsky2018note} demonstrate, no best method is best across all settings, and listwise deletion can outperform alternatives (e.g., multiple imputation) depending on the underlying data generating process. Our results provide additional support for the perspective that the most suitable inferential strategy is one chosen based on the specifics of the problem at hand.

\section*{Acknowledgments}

The authors thank Forrest Crawford, Natalie Hernandez, Rosa Kleinman, Fredrik S\"avje, as well as the editor, Jeff Gill, and two anonymous referees for helpful comments.

\section*{Data Availability Statement}

Data and code to replicate all simulations and numerical illustrations are available at \citet{dataverse}.

\clearpage

\section*{Appendix: Proofs}

\begin{proof}[Proof of Lemma 3]

We will prove the result in two cases. First suppose $q_* = 0$, which is equivalent to say that all variables are fully missing. Then $p_{all} = 1 = (1-q_*^k)^n$. Now suppose $q_* \in (0,1)$. By the independence assumption, $p_{all}=\Pi_{i=1}^n (1-\Pr (\M_{ik} = \mathbf{0}))$. Denote $q_{ij} = \Pr( M_{ij} = 0 | \M_{i(j-1)} = \mathbf{0})$ if $j>1$, else $q_{ij} =  \Pr( M_{ij} = 0)$. By Assumption 2, $q_{ij} \le q_*$, for all $j \in \{1,2,\ldots,k\}$. By the chain rule of conditional probability, $\Pr  (\M_{ik} = \mathbf{0}) = q_{i1} q_{i2} \cdots q_{ik}.$ 
This means that the probability of a single observation containing at least one missing entry is $(1-q_{i1} q_{i2} \cdots q_{ik})$. Since $q_* \ge q_{ij}$ for all $j \in \{1,2,\ldots,k\},$ $q_*^k \ge q_{i1} q_{i2} \cdots q_{ik}$. Thus $(1-q_*^k) \le (1-q_{i1}q_{i2} \cdots q_{ik})$.  Thus  $(1-q_*^k)^n$ is a lower bound for the probability of all $n$ observations each containing at least one missing entry.
\end{proof}

\begin{proof}[Proof of Proposition 5]
First we will show that $\lim_{n \to \infty} nq_*^{f(n)}=0$ (in asymptotic shorthand notation, $q_*^{f(n)} = o(n^{-1})$). Note that
\[ \lim_{n \to \infty} nq_*^{f(n)} = \lim_{n \to \infty} e^{\log nq_*^{f(n)}} =  \lim_{n \to \infty} e^{\log n + f(n) \log q_*}.\]
Since $q_* \in (0,1)$, $\log q_* < 0$. Since $f(n) = \omega(\log n)$,  the sequence $\log n + f(n) \log q_*$ diverges to negative infinity, and so \[\lim_{n \to \infty} e^{\log n + f(n) \log q_*}=0 =  \lim_{n \to \infty} nq_*^{f(n)}.\] 

Since $q_* \in (0,1)$ and $k=f(n) \geq_* 1$, $-q_*^{f(n)} > -1$ and $1-q_*^{f(n)} \leq 1$. By Bernoulli's Inequality, since $n \in \mathbb{N}, (1-q_*^{f(n)})^n \geq_* 1+n(-q_*^{f(n)})=1-nq_*^{f(n)}$. Thus $1-nq_*^{f(n)} \leq (1-q_*^{f(n})^n \leq 1$ in the common domain $n \in \mathbb{N}$. Since
 $\lim_{n \to \infty} 1 = 1$ and $\lim_{n \to \infty} 1- nq_*^{f(n)} = 1- \lim_{n \to \infty} nq_*^{f(n)} = 1$, by the Squeeze Theorem,
\[\lim_{n \to \infty} (1-q_*^{f(n)})^n=1.\]
Then, since $\forall n, (1-q_*^{f(n)})^n \leq p_{all} \leq 1$, we have $\lim_{n \to \infty} p_{all} = 1$, again by the Squeeze Theorem. 
\end{proof}

\clearpage

\section*{Supporting Information A: Numerical Illustration of Theoretical Findings}

Here we demonstrate the finite-$n$ implications of our findings. In the following figures, we plot figures for different values of an upper bound on the conditional probability of each data entry being observed $q_*$, the number of observations (rows) $n$, the number of variables (columns) $k$ and a sharp lower bound on $p_{all}$ denoted $\underline{p_{all}} = (1-q_*^{f(n)})^n$. $\underline{p_{all}}$ is the lowest possible ("best-case") probability that listwise deletion removes all data. 

Figure 3 demonstrates how listwise deletion asymptotically removes all data. Figure 3 plots the sharp lower bound $\underline{p_{all}}$ for the probability that listwise deletion removes all data against the number of variables $k$ in three different settings for $n=100,1000,10000$ in each subfigure. By Lemma 3, we can compute $\underline{p_{all}} = (1-q_*^k)^n$.

\begin{figure}[h]
\centering
\includegraphics[width=4.5in]{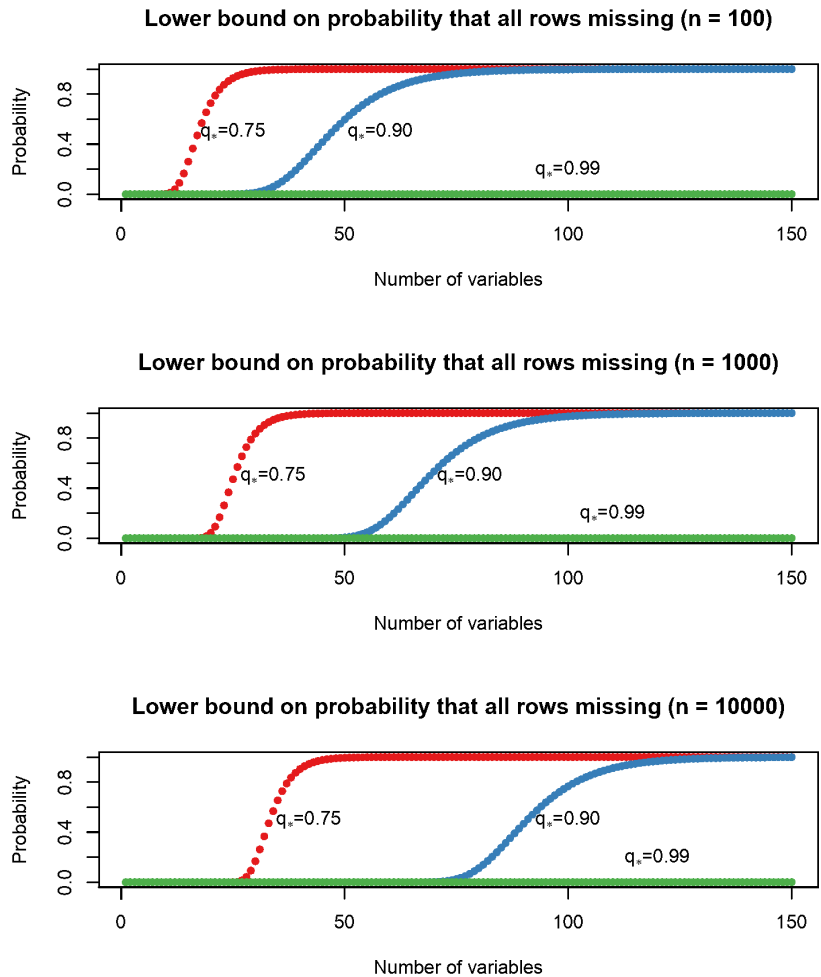}
\caption{Lower bound of probability that all rows missing ($\underline{p_{all}}$) plotted against values of $n$,$k$ and $q_*$.}
\end{figure}

In each subfigure, we simultaneously consider three different values for the upper bound of the conditional probability $q_*$ defined in Assumption \ref{qbound}: the red  curves represent $\underline{p_{all}}$ calculated with $q_*=0.75$, the blue curves represent $\underline{p_{all}}$ calculated with $q_*=0.90$ , and the green curves represent $\underline{p_{all}}$ calculated with $q_*=0.99$.   We see that the rate of $\underline{p_{all}}$ converges to $1$ as $k$ gets large, which is faster when $q$ gets smaller. However, the rate of convergence heavily depends on $q_*$. When $q_* = 0.75$ (i.e., each variable has at least a 25\% chance of idiosyncratic missingness), the lower bound $\underline p_{all}$ is extremely close to $1$ even when $n=10,000$. However, when the upper bound $q_*$ is as large as $0.99$, $\underline{p_{all}}$ is essentially zero when $n=100$ and $k=150$, reflecting the fact that the probability of idiosyncratic missingness is essential in determining the properties of listwise deletion.

Figure 4 illustrates, for a given $n, q_*$, how large  $k$ can be while still ensuring that  $p_{all} \leq 0.5,0.99$. We compute this using the result from Lemma 3, $\underline{p_{all}} = (1-q_*^k)^n$. Since $(1-q_*^k)^n$ is strictly increasing in $k$, solving for equality $p_{all} = (1-q_*^k)^n$ we will get the smallest possible $k$ for each $\underline{p_{all}}$ that $k = \bigg\lfloor{\frac{\log ( 1-\underline{p_{all}}^{1/n})}{\log(q_*)}}\bigg\rfloor. $ We present two subfigures:   with $\underline{p_{all}} =0.5$ for the first subfigure and $\underline{p_{all}} =0.99$ for the second subfigure, and we plot the $k$ against the $n$ for three different upper bounds $q_*=0.75,0.90,0.99.$ 

With more missingness, $q_*=0.75,0.90$, even relatively small $k$ can yield missingness of $\underline{p_{all}}$. For example, even with $n=10,000$ and $q_*=0.75$  we need only $k=33$ to have a 50\% probability that all rows will be missing. However, when missingness is very low, $k$ needs to be very large to cause all data to be missing. For example,  with $n=10,000$ and  $q_*=0.99$, we need $k=952$ to have a 50\% probability that all rows will be missing.

\begin{figure}[H]
\centering
\begin{subfigure}[b]{0.45\textwidth}
         \centering
\includegraphics[width=\textwidth]{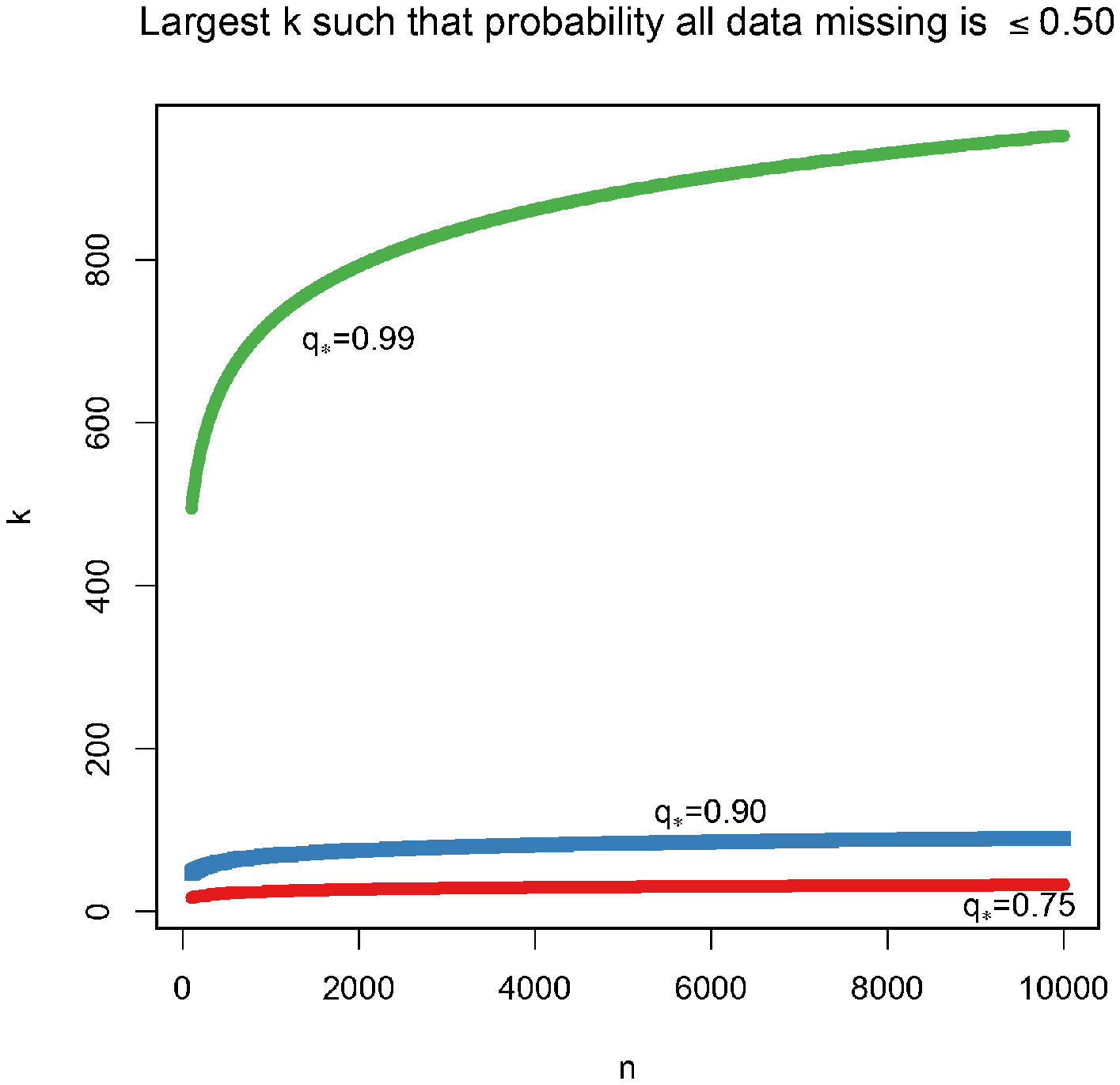}
         \caption{$p_{all} \le 0.5$}
         \label{fig:y equals x}
     \end{subfigure}
     \hfill
     \begin{subfigure}[b]{0.45\textwidth}
         \centering
 \includegraphics[width=\textwidth]{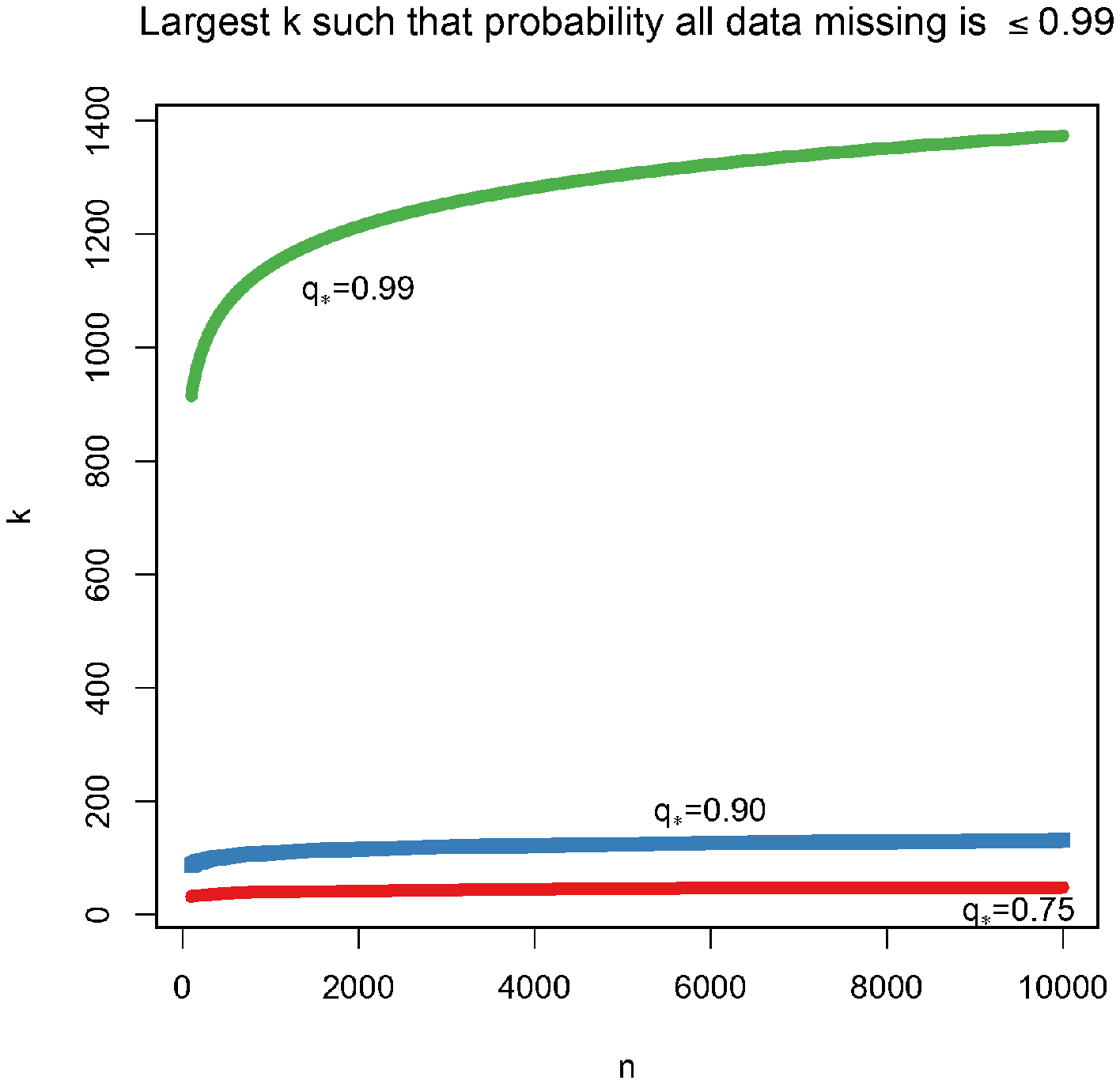}
         \caption{$p_{all} \le 0.99$}
         \label{fig:three sin x}
     \end{subfigure}
\caption{Largest $k$ such that  $p_{all} \leq 0.5,0.99$ plotted against $n$ and $q_*$.}
\end{figure}

Our final  numerical illustration considers an upper bound on the expected proportion of observations that are missing, $1-q_*^k$, which does not depend on $n$. Figure 5 plots  the expected proportion of data missing versus the number of variables $k$. We see the same qualititative relationship as before --- as the number of variables increases, we have a very quick decline in the proportion of usable data.   In comparison to Figure 3, the expected proportion of data missing tends faster to $1$ for each $q_*$ considered as $k$ gets large, as it is equivalent to the special case for $\underline{p_{all}}$ when $n=1$. 

\begin{figure}[H]
\centering
\includegraphics[width=4in]{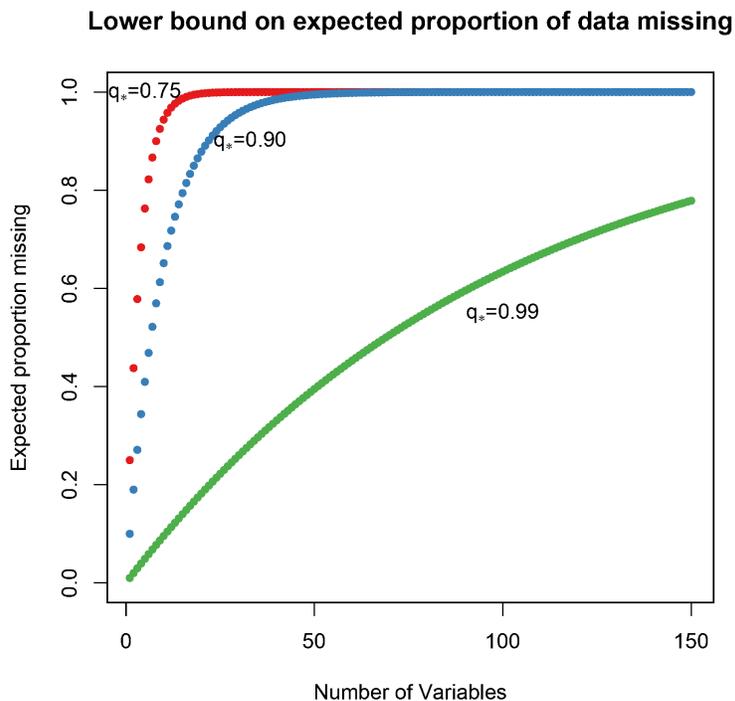}
\caption{Lower bound of expected proportion of all rows missing ($\underline{p_{all}}$) plotted against values of $n$,$k$ and $q_*$.}
\end{figure}

\clearpage

\section*{Supporting Information B: Asymptotics in the Number of Groups of Variables}

In this section, we provide a formal exposition of how our results can generalize to the case where we have idiosyncratic missingness with respect to groups of variables, rather than each specific variable. The language is largely duplicative of the language in Section 2 of the main text; however it makes explicit the direct manner in which the result can generalize.

Let $n$ be the number of observations. Let $g$ be the number of variable groups in the dataset, within which all observations share an identical missingness pattern. We let $M_{ij}$ be a random indicator variable for whether or not the $j$th group in the $i$th row is missing. We use one indicator for each group due to the shared missingness. Similar to the previous setting, we use $\M_{ij}$ to represent the random vector collecting the missingness indicators up to variable $j$, $(M_{i1},M_{i2},\ldots,M_{ij})$. Let $k$ be the number of variables in the dataset. Note that, by construction, we know that $k \geq g$ since groups contain at least one variable.

We will restate Assumption 1 and 2 in the group settings in Assumption 6 and 7 such that there is mutual independence of missingness across rows as well as the  conditional probability that an observation is missing being bounded away from zero.

\begin{assumption} \label{independence*}
All rows of the data $\left( (M_{11},...,M_{1g} ), ..., (M_{n1},...,M_{ng} )  \right) $ are mutually independent.
\end{assumption}

\begin{assumption}\label{qbound*}
There exists a $q_* \in [0,1)$ such that for all $i$,
\begin{itemize}
\item $\Pr( M_{i1} = 0) \leq q_*$
\item 
$\Pr( M_{ij} = 0 | \M_{i(j-1)} = \mathbf{0} ) \le q_*$, for all $j \in \{2,\ldots,g\}$ such that $\Pr(\M_{i(j-1)} = \mathbf{0}) > 0$ 
\end{itemize}
\end{assumption}

Then we can obtain a group version of Lemma 3 (Lemma 8) following similar steps.

\begin{lemma} Under Assumptions \ref{independence*} and \ref{qbound*}, the probability that listwise deletion removes all rows is $p_{all} \geq (1-q_*^g)^n$. 
\end{lemma}

\begin{proof}[Proof of Lemma 8]

Similar to the proof of Lemma 3, we will still consider two cases for $q_*$. Suppose $q_*=0$. Since this entails that all groups of variables are completely missing, $p_{all} = 1 = (1-q_*^g)^n$. For the second case suppose $q_* \in (0,1)$. By the group independence assumption, $p_{all}=\Pi_{i=1}^n (1-\Pr (\M_{ig} = \mathbf{0}))$. Denote $q_{ij} = \Pr( M_{ij} = 0 | \M_{i(j-1)} = \mathbf{0})$ if $j>1$, else $q_{ij} =  \Pr( M_{ij} = 0)$. By Assumption 7, $q_{ij} \le q_*$, for all $j \in \{1,2,\ldots,g\}$. By the chain rule of conditional probability, $\Pr  (\M_{ig} = \mathbf{0}) = q_{i1} q_{i2} \cdots q_{ig}.$ 
This means that the probability of a single observation containing at least one missing entry is $(1-q_{i1} q_{i2} \cdots q_{ig})$. Since $q_* \ge q_{ij}$ for all $j \in \{1,2,\ldots,g\},$ $q_*^g \ge q_{i1} q_{i2} \cdots q_{ig}$. Thus $(1-q_*^g) \le (1-q_{i1}q_{i2} \cdots q_{ig})$.  Thus  $(1-q_*^g)^n$ is a lower bound for the probability of all $n$ observations each containing at least one missing entry.
\end{proof}

Similarly, we will embed the problem into a sequence $g_n=l(n)$, where $l$ has range over the natural numbers, and allow $M_{ij,n}$ and $q_{ij,n}$ to vary at each $n$. We omit the $n$ notation for simplicity. Hence we have the third assumption in the group setting that $g$ grows superlogarithmically in  $n$. We discussed the interpretation of this assumption in the variable setting extensively in the Theory section, so here we will only present the assumption and the group version of Proposition 5 in Proposition 10, as well as a proof for Proposition 10.

\begin{assumption}\label{superlog*}

The number of groups of covariates grows superlogarithmically in $n$, so that $\lim_{n\rightarrow\infty} \frac{l(n)}{log(n)} = \infty$. 

\end{assumption}

\begin{prop}\label{prop*}  Under Assumptions \ref{independence*}, \ref{qbound*} and \ref{superlog*}, $\lim_{n \to \infty} p_{all} = 1$.
\end{prop}

\begin{proof}[Proof of Proposition 10]
First we will show that $\lim_{n \to \infty} nq_*^{l(n)}=0$ (in asymptotic shorthand notation, $q_*^{l(n)} = o(n^{-1})$). Note that
\[ \lim_{n \to \infty} nq_*^{l(n)} = \lim_{n \to \infty} e^{\log nq_*^{l(n)}} =  \lim_{n \to \infty} e^{\log n + l(n) \log q_*}.\]
Since $q_* \in (0,1)$, $\log q_* < 0$. Since $l(n) = \omega(\log n)$,  the sequence $\log n + l(n) \log q_*$ diverges to negative infinity, and so \[\lim_{n \to \infty} e^{\log n + l(n) \log q_*}=0 =  \lim_{n \to \infty} nq_*^{l(n)}.\] 

Since $q_* \in (0,1)$ and $k=l(n) \geq_* 1$, $-q_*^{l(n)} > -1$ and $1-q_*^{l(n)} \leq 1$. By Bernoulli's Inequality, since $n \in \mathbb{N}, (1-q_*^{l(n)})^n \geq_* 1+n(-q_*^{l(n)})=1-nq_*^{l(n)}$. Thus $1-nq_*^{l(n)} \leq (1-q_*^{l(n})^n \leq 1$ in the common domain $n \in \mathbb{N}$. Since
 $\lim_{n \to \infty} 1 = 1$ and $\lim_{n \to \infty} 1- nq_*^{l(n)} = 1- \lim_{n \to \infty} nq_*^{l(n)} = 1$, by the Squeeze Theorem,
\[\lim_{n \to \infty} (1-q_*^{l(n)})^n=1.\]
Then, since $\forall n, (1-q_*^{l(n)})^n \leq p_{all} \leq 1$, we have $\lim_{n \to \infty} p_{all} = 1$, again by the Squeeze Theorem. 
\end{proof}

\begin{thebibliography}{}

\bibitem[\protect\citeauthoryear{Allison}{Allison}{2001}]{Allison2001}
Allison, P.~D. (2001).
\newblock {\em Missing Data}.
\newblock Sage University Papers Series on Quantitative Applications in Social
  Sciences. Thousand Oaks, CA: Sage.

\bibitem[\protect\citeauthoryear{Arel-Bundock and Pelc}{Arel-Bundock and
  Pelc}{2018}]{arel2018can}
Arel-Bundock, V. and K.~J. Pelc (2018).
\newblock When can multiple imputation improve regression estimates?
\newblock {\em Political Analysis\/}~{\em 26\/}(2), 240--245.

\bibitem[\protect\citeauthoryear{Berk}{Berk}{1983}]{berk1983handbook}
Berk, R. (1983).
\newblock Applications of the general linear model to survey data.
\newblock In A.~B.~A. Peter H.~Rossi, James D~Wright (Ed.), {\em Handbook of
  Survey Research}, Quantitative Studies in Social Relations. New York:
  Academic Press.

\bibitem[\protect\citeauthoryear{Cameron and Trivedi}{Cameron and
  Trivedi}{2005}]{cameron2005microeconometrics}
Cameron, A. and P.~Trivedi (2005).
\newblock {\em Microeconometrics: Methods and Applications}.
\newblock New York: Cambridge University Press.

\bibitem[\protect\citeauthoryear{Esty, Goldstone, Gurr, Harff, Levy, Dabelko,
  Surko, and Unger}{Esty et~al.}{1999}]{sfphase2}
Esty, D., J.~Goldstone, T.~Gurr, B.~Harff, M.~Levy, G.~Dabelko, P.~Surko, and
  A.~Unger (1999).
\newblock State failure task force report: Phase ii findings.
\newblock {\em Environmental Change and Security Project Report\/}~{\em 5},
  49--72.

\bibitem[\protect\citeauthoryear{Esty, Goldstone, Gurr, Surko, and Unger}{Esty
  et~al.}{1995}]{sfphase1}
Esty, D.~C., J.~Goldstone, T.~R. Gurr, P.~Surko, and A.~Unger (1995).
\newblock {\em Working Papers: State Failure Task Force Report}.
\newblock McLean, VA: Science Applications International Corporation.

\bibitem[\protect\citeauthoryear{Friedman, Hastie, and Tibshirani}{Friedman
  et~al.}{2010}]{glmnet}
Friedman, J., T.~Hastie, and R.~Tibshirani (2010).
\newblock Regularization paths for generalized linear models via coordinate
  descent.
\newblock {\em Journal of Statistical Software\/}~{\em 33\/}(1), 1--22.

\bibitem[\protect\citeauthoryear{Honaker and King}{Honaker and
  King}{2010}]{honaker2010missing}
Honaker, J. and G.~King (2010).
\newblock What to do about missing values in time-series cross-section data.
\newblock {\em American journal of political science\/}~{\em 54\/}(2),
  561--581.

\bibitem[\protect\citeauthoryear{King, Honaker, Joseph, and Scheve}{King
  et~al.}{2001}]{king2001analyzing}
King, G., J.~Honaker, A.~Joseph, and K.~Scheve (2001).
\newblock Analyzing incomplete political science data: An alternative algorithm
  for multiple imputation.
\newblock {\em American political science review\/}, 49--69.

\bibitem[\protect\citeauthoryear{King and Zeng}{King and Zeng}{2001}]{kingzeng}
King, G. and L.~Zeng (2001).
\newblock Improving forecasts of state failure.
\newblock {\em World Politics\/}~{\em 53\/}(4), 623--658.

\bibitem[\protect\citeauthoryear{King and Zeng}{King and
  Zeng}{2007}]{king2007dataverse}
King, G. and L.~Zeng (2007).
\newblock {\em Replication data for: Improving Forecasts of State Failure}.
\newblock Harvard Dataverse.

\bibitem[\protect\citeauthoryear{Lai, Robbins, and Wei}{Lai
  et~al.}{1978}]{lai1979strong}
Lai, T.~L., H.~Robbins, and C.~Z. Wei (1978).
\newblock Strong consistency of least squares estimates in multiple regression.
\newblock {\em Proceedings of the National Academy of Sciences of the United
  States of America\/}~{\em 75\/}(7), 3034--3036.

\bibitem[\protect\citeauthoryear{Lall}{Lall}{2016}]{lall2016multiple}
Lall, R. (2016).
\newblock How multiple imputation makes a difference.
\newblock {\em Political Analysis\/}~{\em 24\/}(4), 414--433.

\bibitem[\protect\citeauthoryear{Lehmann}{Lehmann}{2006}]{lehmann2006elements}
Lehmann, E. (2006).
\newblock {\em Elements of Large-Sample Theory}.
\newblock Springer Texts in Statistics. Springer New York.

\bibitem[\protect\citeauthoryear{Little and Rubin}{Little and
  Rubin}{2019}]{little2019statistical}
Little, R.~J. and D.~B. Rubin (2019).
\newblock {\em Statistical analysis with missing data}, Volume 793.
\newblock John Wiley \& Sons.

\bibitem[\protect\citeauthoryear{Liu, Wang, Feng, , and Wall}{Liu
  et~al.}{2016}]{liu}
Liu, Y., Y.~Wang, Y.~Feng, , and M.~M. Wall (2016, Mar).
\newblock Variable selection and prediction with incomplete high-dimensional
  data.
\newblock {\em Ann Appl Stat\/}~{\em 10\/}(1), 418--450.

\bibitem[\protect\citeauthoryear{Pepinsky}{Pepinsky}{2018}]{pepinsky2018note}
Pepinsky, T.~B. (2018).
\newblock A note on listwise deletion versus multiple imputation.
\newblock {\em Political Analysis\/}~{\em 26\/}(4), 480--488.

\bibitem[\protect\citeauthoryear{{R Core Team}}{{R Core Team}}{2020}]{r}
{R Core Team} (2020).
\newblock {\em R: A Language and Environment for Statistical Computing}.
\newblock Vienna, Austria: R Foundation for Statistical Computing.

\bibitem[\protect\citeauthoryear{Schafer}{Schafer}{1997}]{schafer1997}
Schafer, J.~L. (1997).
\newblock {\em Analysis of Incomplete Multivariate Data}.
\newblock Chapman \& Hall/CRC.

\bibitem[\protect\citeauthoryear{Stata.com}{Stata.com}{2020}]{regress}
Stata.com (2020).
\newblock Regress - linear regression.
\newblock https://www.stata.com/manuals13/rregress.pdf.

\bibitem[\protect\citeauthoryear{Teorell, Sundström, Holmberg, Rothstein,
  Pachon, and Dalli}{Teorell et~al.}{2021}]{qog2021}
Teorell, J., A.~Sundström, S.~Holmberg, B.~Rothstein, N.~A. Pachon, and C.~M.
  Dalli (2021).
\newblock {\em The Quality of Government Standard Dataset version Jan21}.
\newblock University of Gothenburg: The Quality of Government Institute.



\bibitem[\protect\citeauthoryear{Wang, J. Sophia and Peter M. Aronow}{Wang and Aronow}{2021}]{dataverse}
Wang, J.~S. and P.~M. Aronow (2021). Replication Data for: Listwise Deletion In High Dimensions. https://doi.org/10.7910/DVN/T8BG2K, Harvard Dataverse, DRAFT VERSION, UNF:6:0gB5c9RyKb6AH1zMEUNOpQ== [fileUNF] 


\end{thebibliography}
\end{document}